\documentclass[12pt,fleqn,thmsa]{article}
\usepackage{amsfonts,amsthm}
\usepackage{amssymb,hyperref,backref}
\usepackage{amsmath,harvard,tikz,verbatim,epsfig}
\usepackage{xcolor}
\usepackage{graphicx}
\usepackage{scalefnt}
\usepackage[bottom]{footmisc}

\newcommand{\be}{\begin{eqnarray*}}
\newcommand{\ee}{\end{eqnarray*}}

\newtheorem{theorem}{Theorem}[section]

{ 
}

\begin{document}

\title{{\large Endogeneity in Weakly Separable Models without Monotonicity}\thanks{
    An earler version of this paper was circulated under the title ``Identification and Estimation of weakly separable models without Monotonicity". We are grateful to Jeremy Fox, Sukjin Han, Arthur Lewbel, Whitney Newey, Elie Tamer, Ed Vytlacil, Haiqing Xu, conference participants at the 2018 West Indies Economic Conference, the 2019 SUFE econometrics meetings, the 2020 Texas Camp Econometrics,  the 2020 World Congress, and seminar participants at various institutions for helpful comments and suggestions. }}
\author{{\small Songnian Chen} \\
{\small Zhejiang University} \and {\small Shakeeb Khan} \\
{\small Boston College} \and {\small Xun Tang} \\
{\small Rice University }}
\date{{\footnotesize \today}}
\maketitle

\begin{center}
{\footnotesize \textbf{Abstract} }
\end{center}

\noindent {\footnotesize 

We identify and estimate treatment effects when potential outcomes are weakly separable with a binary endogenous treatment. 
\citeasnoun{vytlacilyildiz} proposed an identification strategy that exploits the mean of observed outcomes, but their approach requires a monotonicity condition. 
In comparison, we exploit full information in the entire outcome distribution, instead of just its mean. As a result, our method does not require monotonicity and is also applicable to general settings with multiple indices. We provide examples where our approach can identify treatment effect parameters of interest whereas existing methods would fail. These include models where potential outcomes depend on multiple unobserved disturbance terms, such as a Roy model, a multinomial choice model, as well as a model with endogenous random coefficients. We establish consistency and asymptotic normality of our estimators.}

\noindent {\small \textbf{JEL Classification}: C14, C31, C35 \newline
} \pagebreak

\setcounter{equation}{0}

\section{Introduction}

\label{introduction}

Consider a weakly separable model with a binary endogenous variable:%
\begin{eqnarray}
Y &=&g(v_1(X,D),v_2(X,D),...v_J(X,D),\varepsilon )  \label{outcomeeq} \\
D &=&1\left \{ \theta (Z)-U>0\right \} \label{treateq}
\end{eqnarray}%
where $( v_1(X,D),v_2(X,D),...v_J(X,D) ) \equiv v(X,D) $ is a $J$-vector of
unknown linear or nonlinear indices in the outcome equation (1.1) and $D$ is
a binary endogenous variable defined by equation (1.2). Here $%
X\in \mathbb{R}^{d_{x}}$ and $Z\in \mathbb{R}^{d_{z}}$ are vectors of
observable exogenous variables, which may have overlapping elements. This paper is about the identification and estimation of the
average treatment effect (ATE).\footnote{%
The ATE is one of several parameters of interest in the program evaluation
literature. In their seminal work, \citeasnoun{imbensangrist} focus on a
different parameter of interest, the Local Average Treatment Effect (LATE)
which is the ATE for the subset of the population referred to as compliers.
However, as pointed out in \citeasnoun{HV05}, \citeasnoun{heckvyt-handbook}, %
\citeasnoun{CHV10}, and \citeasnoun{santos18}, there are settings where the
LATE by itself is not always the focal point of policy studies.}  Similar to %
conditions in \citeasnoun{vytlacilyildiz}, we require that there is some continuous element in $Z$
excluded from $X$ so that $\theta(Z)$ varies continuously conditional on $X$, and that we can vary $X$ after conditioning on $\theta(Z)$.%
\footnote{%
This does not necessarily require a second exclusion restriction that there be an
element in $X$ not in $Z$. As explained in \citeasnoun{vytlacilyildiz}, what
is required is that the element in $Z$ not in $X$ be continuously
distributed.}
It is worth noting that the method we propose also applies directly in a more general setup where the error in the outcome equation is specific to $D$, i.e., with $\varepsilon$ replaced by $\varepsilon_D$, provided the marginal distribution of $\varepsilon_d$ is the same for $d\in\{0,1\}$.

The literature on program evaluation abounds in examples in which self-selected, endogenous treatment depends on exogenous instruments that have no direct impact on potential outcomes.
For instance, \citeasnoun{vytlacilyildiz} provided an example where $D$ indicates an individual's enrollment in job training at date 0, and $Y$ is employment status at date 1. In this case, $Z$ may contain variables summarizing local labor market conditions at date 0. They noted these variables affect the opportunity cost of training, and consequently the training decision at date 0, but do not directly affect the individual's employment at date 1.

In the system of equations above, $U$ is the unobservable
random variable normalized to follow the standard uniform distribution, and
the error term $\varepsilon$ in the outcome equation is allowed to be a
random vector with a known dimension.  We assume $\left( X,Z\right) $ are independent of $%
(\varepsilon ,U)$. Note that we allow $v(X,D)$ to be a vector of multiple
indices, whereas existing methods, such as \citeasnoun{vytlacilyildiz}, can only be
applied when it is a single index.
\citeasnoun{HV05} and \citeasnoun{carneiroLee2009} maintained that $(\varepsilon,U)$ is independent of $Z$ conditional on $X$, and used local instrumental variables to estimate marginal treatment effect (MTE) for realized values of $U$ over the support of $\theta(Z)$ given $X$. In their case, integrating out MTE to ATE would require  the support of $\theta(Z)$ given $X$ to be the complete unit interval, i.e. $[0,1]$. 
In comparison, we propose a method that relaxes such a support condition by exploiting more structure in the potential outcome equation and its  full distribution. 
Moreover, while \citeasnoun{carneiroLee2009} use $E[1\{Y\leq y\}|X,D,U=p]$ for \textit{each fixed $y$} to recover the conditional distribution of potential outcomes at $y$, we use the full distribution of the observed outcomes to find the matching covariates. This helps us to recover ATE without the aforementioned full support condition.
As in \citeasnoun{vytlacilyildiz}, our method requires that at least some covariates in the potential outcome equation are exogenous.

Since \citeasnoun{vytlacilyildiz}, important work has considered
identification and estimation of similar models, but under alternative
conditions, notably on the supports of $D$, and $\theta(Z)$ as well
as the dimension of $\varepsilon$. \citeasnoun{neweyimbens} assume $D$ is
continuous and monotone in the error term, $\theta(Z)$ is continuous with
large support, and $\varepsilon$ is scalar. \citeasnoun{kasy14} allows $%
\varepsilon$ to be a vector as we do here but imposes that both $D$ and $%
\theta(Z)$ are continuously distributed and further assumes a monotone
relationship between the two. \citeasnoun{dhault2014} and %
\citeasnoun{torgo2014} assume $D$ is continuous, $\varepsilon$ is scalar,
and assume monotonicity in the error term in each of the two equations,
though they allow for $\theta(Z)$ to be discrete. \citeasnoun{vuongxu}
assume $D$ is discrete and $\theta(Z)$ is continuous, restricts $\varepsilon$
to be scalar and requires $Y$ in (\ref{outcomeeq}) to be monotone in $%
\varepsilon$. \citeasnoun{JunEtal2016} study an extension of %
\citeasnoun{vytlacilyildiz}, and also require the same monotonicity
condition as in the latter. \citeasnoun{fengjun} shows how
to identify nonseparable triangular models where the endogenous variable is
discrete but restricted to have larger support than the instrument variable.%
\footnote{%
These papers focus on point identification. For partial identification of a
model with a binary outcome, see \citeasnoun{shaikhvytlacil11} and %
\citeasnoun{mourifie15}. \citeasnoun{santos18} explore partial
identification of relevant treatment effect parameters in models without
structure imposed on the outcome equation. In such settings, attaining point
identification requires support conditions on the propensity score function
that are stronger than imposed here.}


As in the conventional framework, two potential outcomes $Y_{1}$ and $Y_{0}$
satisfy%
\begin{equation*}
Y_D = g(v(X,D),\varepsilon) \mbox{ for } D=0,1. \text{\ }
\end{equation*}%
We only observe $\left( Y,D,X,Z\right) $, where $Y=DY_{1}+(1-D)Y_{0}$. In this model, we do not impose parametric distribution on the error terms $(\varepsilon,U)$ or a linear index structure. %
Similarly, \citeasnoun{vytlacilyildiz} also do not require such restrictions, but assume that $v(X,D) \in \mathbb{R}$ is a single
index, and 
\begin{equation}
E\left[ g(v,\varepsilon )|U=u\right] \mbox{ is strictly increasing in }v \in 
\mathbb{R}\mbox{ for all }u.
\end{equation}%
We do not impose any such monotonicity structure. That is because our approach exploits the full information in the distribution of the outcome variable, instead of just its mean. 
Indeed, when the outcome distribution function is more informative than the mean, our method is applicable to more general settings; 
in particular, not only do we not rely on such a monotonicity assumption, but we also allow for multiple indices.

In Section \ref{Identification} we present the identification argument and discuss its required conditions.
In Sections \ref{Examples} and \ref{Extension} we provide some examples in
which such a monotonicity condition fails, but the average effect of the
binary endogenous variable is still identified. To reiterate, we allow the potential outcomes to be weakly separable in multiple indices, that is, $%
v(X,D)=(v_{1}(X,D),v_{2}(X,D),...,v_{J}(X,D)) \in \mathbb{R}^J$.
We consider the identification and estimation of the average treatment
effect of $D$ on $Y$, $E(Y_{1}|X\in A)$, $E(Y_{0}|X\in A)$, and $%
E(Y_{1}-Y_{0}|X\in A)$, for some set $A$, without the aforementioned
monotonicity. Indeed, for the case with multiple indices $v(X,D) \in \mathbb{%
R}^J$, the monotonicity condition is no longer well defined.


\setcounter{equation}{0}

\section{Identification\label{Identification}}

Our identification strategy is based on the notion
of \emph{matching covariates}.
Let $Supp(V)$ denote the support of a generic random vector $V$.
Consider the identification of $E(Y_{1}|X=x)$ for some $x\in Supp(X)$.
Under the assumption that $%
(\varepsilon ,U)\bot (X,Z)$,%
\begin{eqnarray}
&&\text{ }E(Y_{1}|X=x)=E(Y_{1}|X=x,Z=z)  \notag \\
&&\text{ }=E(DY_{1}|X=x,Z=z)+E[(1-D)Y_{1}|X=x,Z=z]  \notag \\
&&\text{ }=P(z)E(Y|D=1,X=x,Z=z)+[1-P(z)]E(Y_{1}|D=0,X=x,Z=z)  \label{EQ1}
\end{eqnarray}%
where $P(z)\equiv E(D|Z=z) = \theta(z)$ because $U$ is normalized to follow a standard uniform distribution. The only term not directly identifiable
on the right-hand side of (\ref{EQ1}) is:\footnote{
	If the support of $P(Z)$ given $x$ covers the full closed interval $[0,1]$, then $E(Y_1|X=x)$ is directly identified as $E(Y|D=1,X=x,Z=z)$ at $z$ s.t. $P(z)=1$. However this means point identification hinges on the event ``$P(Z)=1$ given $x$''. 
	}%
\begin{equation*}
E(Y_{1}|D=0,X=x,Z=z)=E[g(v(x,1),\varepsilon )|U\geq P(z)]\text{.}
\end{equation*}
The main idea behind our approach would at first seem to be similar to that of \citeasnoun{vytlacilyildiz}%
, which is to use exogenous variations of covariates in the outcome equation to 
find some $(\tilde{x},\tilde{z})\in Supp(X,Z)$ such that%
\begin{equation}
P(z)=P(\tilde{z}) \text{ and } v(x,1)=v(\tilde{x},0)  \label{EQ2}
\end{equation}%
so that%
\begin{eqnarray} \label{EQ23}
E(Y|D=0,X=\tilde{x},Z=\tilde{z})&=&E(Y_{0}|D=0,X=\tilde{x},Z=\tilde{z}) \\
&=&E[g(v(\tilde{x},0),\varepsilon )|U\geq P(\tilde{z})]=E[g(v(x,1),\varepsilon
)|U\geq P(z)] \notag \\
&=& E[Y_1|D=0,X=x, Z=z]\text{.} \notag
\end{eqnarray}
However, unlike  \citeasnoun{vytlacilyildiz}, we utilize the full distribution of $Y$ (rather than its first moment only) while searching for such pairs of $(x,%
\tilde{x})$  in (\ref{EQ2}). This allows us to relax the single-index and monotonicity conditions.

To fix ideas, let there be continuous components in $Z$ that are excluded from $X$.
For any $p$ on the support of $P(Z)$ given $X=x$, and for any $y$, define 
\begin{eqnarray}
h_{1}^{\ast }(x,y,p) &=&E[D1\left \{ Y\leq y\right \} |X=x,P(Z)=p]  \notag \\
&=&E\left[ 1\left \{ U<P(Z)\right \} 1\left \{ g(v(X,1),\varepsilon )\leq
y\right \} |X=x,P(Z)=p\right]  \notag \\
&=&\int_{0}^{p}F_{g|u}(y;v(x,1))du\text{,}  \label{EQ3}
\end{eqnarray}%
where 
\begin{equation*}
F_{g|u}(y;v(x,d))\equiv E[1\{g(v(x,d),\varepsilon )\leq y\}|U=u]\text{,}
\end{equation*}%
with $v(x,d)$ being a realized index at $X=x$ and the expectation in the
definition of $F_{g|u}$ is with respect to the distribution of $\varepsilon $
given $U=u$. The last equality in (\ref{EQ3}) holds because of independence
between $(\varepsilon ,U)$ and $(X,Z)$. 
Assume that for all $y,x$ and $d \in \{ 0,1 \}$, the function $F_{g|u}(y;v(x,d))$ is continuous in $u$ over $[0,1]$. This implies $h^*_1$ is differentiable in $p$.
By construction, $h_{1}^{\ast
}(x,y,p)$ is directly identified from the joint distribution of $(D,Y,X,Z)$
in the data-generating process. Furthermore, for any pair $p_{1}>p_{2}$,
define:%
\begin{equation*}
h_{1}(x,y,p_{1},p_{2})\equiv h_{1}^{\ast }(x,y,p_{1})-h_{1}^{\ast
}(x,y,p_{2})=\int_{p_{2}}^{p_{1}}F_{g|u}(y;v(x,1))du\text{.}
\end{equation*}%
Likewise, define 
\begin{eqnarray*}
h_{0}^{\ast }(x,y,p) &=&E((1-D)1\left \{ Y\leq y\right \} |X=x,P(Z)=p) \\
&=&E\left[ 1\left \{ U\geq P(Z)\right \} 1\left \{ g(v(X,0),\varepsilon
)\leq y\right \} |X=x,P(Z)=p\right] \\
&=&\int_{p}^{1}F_{g|u}(y;v(x,0))du.
\end{eqnarray*}%
and let 
\begin{equation*}
h_{0}(x,y,p_{1},p_{2})\equiv h_{0}^{\ast }(x,y,p_{2})-h_{0}^{\ast
}(x,y,p_{1})=\int_{p_{2}}^{p_{1}}F_{g|u}(y;v(x,0))du\text{.}
\end{equation*}%
Let $\mathcal{P}_{x}$ denote the support of $P(Z)$ given $X=x$; let $ Int( \mathcal{P}_{x}\cap \mathcal{P}_{\tilde{x}} )$ denote the interior of the intersection of $\mathcal{P}_{x}$ and $\mathcal{P}_{\tilde{x}}$. Similar to \citeasnoun{HV05} and \citeasnoun{carneiroLee2009}, we use continuous, exogenous variation in $Z$ to establish here that for any $ x \neq \tilde{x}$ with $ Int( \mathcal{P}_{x}\cap \mathcal{P}_{\tilde{x}} ) \neq \emptyset$ and any $y$,%
\begin{equation}
h_{1}(x,y,p,p^{\prime })=h_{0}(\tilde{x},y,p,p^{\prime })\text{ for all }%
p>p^{\prime }\text{ on } Int(\mathcal{P}_{x}\cap \mathcal{P}_{\tilde{x}})  \label{EQ4}
\end{equation}%
if and only if%
\begin{equation}
F_{g|p}(y;v(x,1))=F_{g|p}(y;v(\tilde{x},0))\text{ for all }p\in Int(\mathcal{P}_{x}\cap \mathcal{P}_{\tilde{x}})\text{.}  \label{EQ5}
\end{equation}
\noindent Sufficiency of (\ref{EQ5}) is immediate from the definition of $h_{1}$ and $h_{0}$.
To see its necessity, note that for all $p>p^{\prime }$ on $Int(\mathcal{P}_{x}\cap \mathcal{P}_{%
\tilde{x}})$,%
\begin{equation*}
\left. \frac{\partial h_{1}(x,y,\tilde{p},p^{\prime })}{\partial \tilde{p}}%
\right \vert _{\tilde{p}=p}=\left. \frac{\partial h_{1}^{\ast }(x,y,\tilde{p}%
)}{\partial \tilde{p}}\right \vert _{\tilde{p}=p}=F_{g|p}(y;v(x,1))
\end{equation*}%
{and}%
\begin{equation*}
\left. \frac{\partial h_{0}(\tilde{x},y,\tilde{p},p^{\prime })}{\partial 
\tilde{p}}\right \vert _{\tilde{p}=p}=-\left. \frac{\partial }{\partial 
\tilde{p}}h_{0}^{\ast }(\tilde{x},y,\tilde{p})\right \vert _{\tilde{p}%
=p}=F_{g|p}(y;v(\tilde{x},0))\text{.}
\end{equation*}%
Thus (\ref{EQ4}) and (\ref{EQ5}) are equivalent. 


Next, we collect identifying assumptions as follows:\medskip

\noindent ASSUMPTION A-1: The distribution of $U$ is absolutely continuous
with respect to Lebesgue measure, and is normalized to standard uniform, and $\theta(z)$ in (\ref{treateq}) is continuous in $z$. \medskip

\noindent ASSUMPTION A-2: The random vectors $(U,\varepsilon )$ and $(X,Z)$ are independent.
There exists at least one continuously distributed component in $Z$ that is excluded from $X$.
\medskip

\noindent ASSUMPTION A-3: Both $g(v(X,1),\varepsilon )$ and $g(v(X,0),\varepsilon )$ have finite first moments conditional on $U=u$ for all $u \in [0,1]$. 

\noindent  ASSUMPTION A-4: For all realized values of $y,x$ and $d\in\{0,1\}$, $E[1\{g(v(x,d),\epsilon)
\leq y\}|U=u] $ is continuous in $u$ over $[0,1]$. \medskip

\noindent 
ASSUMPTION A-5: For any $x\neq \tilde{x}$ such that $Int(\mathcal{P}_{x}\cap \mathcal{P}_{\tilde{x}})$ is nonempty, 
$F_{g|p}(y;v(x,1))=F_{g|p}(y;v(\tilde{x},0))$ for all $y$ and $p\in Int(\mathcal{P}_{x}\cap \mathcal{P}_{\tilde{x}})$ if and only if $v(x,1)=v(\tilde{x},0)$.\medskip

Assumptions A-1 and A-3 are common regularity conditions in the literature. Assumption A-2 consists of a standard condition of instrument exogeneity. It allows $P(Z)$ to have continuous, exogenous variation conditional on $X$. The continuity of instrument is also common in the literature on treatment effects. Examples include \citeasnoun{heckmanvytlacil} and \citeasnoun{vytlacilyildiz}.
As noted above, Assumption A-4 ensures the identifiable functions $h^*_1(x,y,p),h^*_0(x,y,p)$ are both differentiable in the last argument $p$.

It is important to note that Assumption A-5 allows the support of $P(Z)$ to be a strict subset of the interval $(0,1)$. 
This is because our method does not use an identification-at-infinity argument, which would require $P(Z)$ to have full support [0,1] given $X$ in order to identify $E(Y_0|X)$ directly. 
(See Footnote 4 above.)
We also note that Assumption A-5 relaxes two limitations of Assumption 4 in \citeasnoun{vytlacilyildiz}. Specifically, to identify pairs $(x,\tilde{x})$ with $v(x,1)=v(\tilde{x},0) $, \citeasnoun{vytlacilyildiz} rely on two assumptions that $v(X,D) \in \mathbb{R}$ is a single index and that $E\left[ g(v(x,1),\varepsilon )|U=u\right] =E\left[ g(v(\tilde{x},0),\varepsilon )|U=u\right] $ if and only if $v(x,1)=v(\tilde{x},0) $. 
The latter holds under a maintained assumption that $E\left[
g(v(x,d),\varepsilon )|U=p\right] $ is a strictly monotonic function of $%
v(x,d)$. 
In comparison, we construct an identification strategy without the single index and monotonicity restrictions by matching conditional distributions $F_{g|p}(\cdot ;v(x,1))$ and $\ F_{g|p}(\cdot ;v(\tilde{x},0))$. 
    
The role of Assumption A-5 in our method can be illustrated by drawing an analogy with a standard identifying condition in nonlinear regression: $ Y = f(X,\theta ^*) + \varepsilon$. In this case, identification of $\theta^*$ requires: ``$f(x,\theta^* ) = f(x,\theta)$ for all $x$ if and only if $\theta^* =\theta$ ''.
To see how this is related to Assumption A-5 in our setting, let $ \theta_1, \theta_0$ be shorthand for $v(x,1), v(\tilde{x},0)$ respectively, and let $m(y,p,\theta) \equiv F_{g|p}(y,\theta )$.
Then our method requires: ``$m(y,p,\theta _{1})=m(y,p,\theta _{0})$ for all $(y,p)$ if and only if $\theta _{1}=\theta _{0}$ ''.

It is worth emphasizing that Assumption A-5 only presents our identification condition in the weakest form, for the sake of generality. 
Later in Section \ref{Examples}, we exploit the structures embedded in specific examples to show how Assumption A-5 can be satisfied  under intuitive, mild primitive conditions.

Next, we specify conditions under which such pairs of covariates exist on the support. Define $\mathcal S \equiv \{(x,\tilde{x}):v(x,1)=v(\tilde{x},0)\} $ and $ \mathcal T \equiv \{(x,\tilde{x}): \exists z, \tilde{z} \text{ with } (x,z),(\tilde{x},\tilde{z})\in Supp(X,Z) \text{ and } P(z)=P(\tilde{z}) \in Int(\mathcal{P}_x\cap \mathcal{P}_{\tilde{x}})\}$.
Let $ \mathcal X^1 \equiv \{x:\exists \tilde{x} \text{ with }(x,\tilde{x})\in \mathcal S \cap \mathcal T\}$ and $ \mathcal X^0 \equiv \{x:\exists \tilde{x} \text{ with }(\tilde{x},x)\in \mathcal S \cap \mathcal T\}$.

\noindent ASSUMPTION A-6: $\Pr (X\in \mathcal X^1)>0$ \ and $\Pr (X\in \mathcal X^0)>0$.\medskip

This condition is similar to Assumption A-4 in \citeasnoun{vytlacilyildiz}. By definition, for each $x\in\mathcal X^1$, we can find $\tilde{x}\in Supp(X)$ such that there exists $(z,\tilde{z}) $ with $ (x,z),(\tilde{x},\tilde{z}) \in Supp(X,Z)$, $v(x,1)=v(\tilde{x},0)$, and $P(z)=P(\tilde{z})$. 
This requires variation in $ x $ while holding $P(Z)$ fixed.
As noted earlier (Footnote 2), this does \textit{not} necessarily require a ``second exclusion restriction'' that there is an element in $X$ that is excluded from $Z$.
In general, there exist multiple such values of $\tilde{x}$ that can be matched with this $x$.
Hence for each $x\in\mathcal X^1$, we define the set of such matched values as 
\begin{equation*}
\lambda_0(x) \equiv \{\tilde{x}: 
h_{1}(x,y,p,p^{\prime })=h_{0}(\tilde{x},y,p,p^{\prime }) \text{ }\forall p>p^{\prime }\text{ on } Int(\mathcal{P}_{x}\cap \mathcal{P}_{\tilde{x}})\},
\end{equation*}
and define $ \mathcal{P}_0^*(x) \equiv  \bigcup_{\tilde{x} \in \lambda_0(x)}(\mathcal{P}_{\tilde{x}}\cap \mathcal{P}_x)$.
Note that by definition of $\mathcal{X}^1$, the set $\mathcal{P}_x \cap \mathcal{P}_{\tilde{x}}$ must be non-empty for $x \in \mathcal{X}^1$ and $ \tilde{x} \in \lambda_0(x)$ under our maintained assumptions.
Likewise, by symmetry, for each $x\in\mathcal X^0$, we define
\begin{equation*}
\lambda_1(x) \equiv \{\tilde{x}:
h_{0}(x,y,p,p^{\prime })=h_{1}(\tilde{x},y,p,p^{\prime }) \text{ }\forall p>p^{\prime }\text{ on } Int(\mathcal{P}_{x}\cap \mathcal{P}_{\tilde{x}})\}, 
\end{equation*}
and let $ \mathcal{P}_1^*(x) \equiv  \bigcup_{\tilde{x} \in \lambda_1(x)}(\mathcal{P}_{\tilde{x}}\cap \mathcal{P}_x)$.

The theorem below shows how to use such matched values to  identify the conditional mean of potential outcomes.

\begin{theorem}
\label{theorem2.1} Suppose Assumptions (A-1)-(A-6) hold. For each $x\in\mathcal X^1$, \begin{equation} \label{id_EY1}
E(Y_1|X=x) = E(DY|X=x,P(Z)=p) + E[(1-D)Y|X\in\lambda_0(x),P(Z)=p]
\end{equation}
for any $p\in \mathcal{P}_0^*(x)$.
For each $x\in\mathcal X^0$,
\begin{equation} \label{id_EY0}
E(Y_0|X=x) = E(DY|X\in\lambda_1(x),P(Z)=p) + E[(1-D)Y|X=x,P(Z)=p]
\end{equation}
for any $p\in \mathcal{P}_1^*(x)$.
\end{theorem}

\begin{proof} As established in the text, (\ref{EQ4}) and (\ref{EQ5}) are equivalent under Assumptions (A-1)-(A-4). 
Then by Assumption (A-5), we know that for each  $x\in\mathcal X^1$ and $\tilde{x}\in \lambda_0(x) $, there exists  $z,\tilde{z}$ such that $(x,z),(\tilde{x},\tilde{z})\in Supp(X,Z)$, $v(x,1)=v(\tilde{x},0)$ and $P(z)=P(\tilde{z})=p\in \mathcal{P}^*_0(x)$.
It follows from (\ref{EQ1}) and (\ref{EQ23}) that
\begin{equation}\label{EQ1Repeat}
E(Y_{1}|X=x)  = E(DY|X=x,Z=z)+E[(1-D)Y|X=\tilde{x},Z=\tilde{z}].  %
\end{equation}%
By (A-2), the right-hand side of (\ref{EQ1Repeat})  equals the right-hand side of (\ref{id_EY1}).
A symmetric argument proves (\ref{id_EY0}).
\end{proof}

It is worth mentioning that the conditions for this theorem do not directly restrict the support of potential outcomes.
As we show in the next section, our method applies in important applications where the potential outcome is either discrete (e.g., determined by multinomial choices), or multi-dimensional with both discrete and continuous components (e.g., determined in a Roy model). 

\setcounter{equation}{0}

\section{Examples}

\label{Examples}

In this section, we present several examples in which the potential outcomes depend on multi-dimensional indices with an endogenous treatment. 
More importantly, we show how the specific structure embedded in each application naturally leads to transparent, primitive conditions that imply Assumption A-5, thus corroborating the wide applicability of this general approach.

In the first and third examples, the monotonicity condition in \citeasnoun{vytlacilyildiz} does not hold; in the second example, the identification requires a generalization of the monotonicity condition into an invertibility condition in higher dimensions.

\noindent \textbf{Example 1. (Heteroskaedastic shocks in Uncensored or Censored outcomes)} First, consider a triangular system where a continuous \textit{uncensored} outcome is determined by double indices $v(X,D)\equiv (v_{1}(X,D),v_{2}(X,D))$:\footnote{
    \citeasnoun{abrevaya2019estimation} adopted a different approach to identify ATE in the same model with uncensored, continuous outcome and multiplicative, heteroskaedastic shocks above, using a binary instrument $Z$. 
    They showed the conditional mean of the observed outcome $Y$, when scaled by the conditional covariance of $Y1\{D=d\}$ and $Z$, is a weighted sum of the conditional means of potential outcomes: $v_1(x,0)$ and $v_1(x,1)$. 
    Thus, using the means of the scaled $Y$ conditional on $Z = 0,1$ respectively, one can construct a linear system that identifies $v_1(x,0)$ and $v_1(x,1)$, provided the proper rank condition holds.
    Their method leverages this particular specification of multiplicative endogeneity, but does not generalize to the case of censored outcomes or in the other examples we consider later.
}
\begin{equation*}
Y=g(v(X,D),\varepsilon )=v_{1}(X,D)+v_{2}(X,D)\varepsilon \text{ for }D\in
\{0,1\} \text{.}
\end{equation*}
The selection equation determining the actual treatment is the same as (1.2). 
In this case the concept of monotonicity in $v\in \mathbb{R}^{2}$ is not well-defined, so the procedure proposed in \citeasnoun{vytlacilyildiz} is not suitable here.\footnote{%
For this particular design, the approach proposed in \citeasnoun{vuongxu}
should be valid. But it will not be for a slightly modified model, such as $%
Y=v_1(X,D)+(e_2+v_2(X,D)*e_1)$, whereas ours will be.} Nevertheless, we can
apply the method in Section \ref{Identification}\ to identify the average
treatment effect by using the \textit{distribution} of outcomes to find pairs
of $x$ and $\tilde{x}$ such that $v(x,1)=v(\tilde{x},0)$. 

To see how Assumption A-5 holds, assume the range of $v_2(\cdot)$ is positive.
Note that%
\begin{eqnarray*}
F_{g|u}(y;v(x,d)) &=&\Pr\left[ v_{1}(x,d)+v_{2}(x,d)\varepsilon \leq y|U=u%
\right] \\
&=&F_{\varepsilon |u}\left( \frac{y-v_{1}(x,d)}{v_{2}(x,d)}\right)
\end{eqnarray*}%
for $d=0,1$. Suppose the distribution of $\varepsilon $ conditional on $u$ is increasing over $\mathbb{R}$.
Then for all $y$ and $x \in \mathcal{X}^1$ and matched values $\tilde{x} \in \lambda_0(x)$, 
\begin{equation*}
F_{g|u}(y;v(x,1))=F_{g|u}(y;v(\tilde{x},0))
\text{ if and only if }%
\frac{y-v_{1}(x,1)}{v_{2}(x,1)}=\frac{y-v_{1}(\tilde{x},0)}{v_{2}(\tilde{x}%
,0)}\text{.}
\end{equation*}%
Differentiating with respect to $y$ yields $v_{2}(x,1)=v_{2}(\tilde{x},0)$, %
which implies $v_{1}(x,1)=v_{1}(\tilde{x},0)$.%

Next, consider a similar model with continuous \textit{censored} outcomes, where
\begin{equation*}
Y=g(v(X,D),\varepsilon )=max\{v_{1}(X,D)+v_{2}(X,D)\varepsilon,0\} \text{ for } D\in
\{0,1\}.
\end{equation*}
In this case, our identification argument for the uncensored case above applies almost immediately. 
The only adjustment needed is to confine the values of $y$ used for constructing matched pairs $(x,\tilde{x})$ over the \textit{uncensored} segment, i.e. $y > 0$. 
In contrast, other methods for identifying ATE in the uncensored model which rely strictly on the linear structure with multiplicative heterogeneity, such as \citeasnoun{abrevaya2019estimation}, no longer applies. 

\noindent \textbf{Example 2. (Multinomial potential outcome)} Consider a
triangular system where the outcome is multinomial. The multinomial response
model has a long and rich history in both applied and theoretical
econometrics. Recent examples in the semiparametric literature include %
\citeasnoun{lee95}, \citeasnoun{pakesporter2014}, \citeasnoun{ahnetal2018}, \citeasnoun{shietal2018}, %
and \citeasnoun{KOT2019}. None of those papers allow for dummy endogenous variables or potential outcomes.

In this example, let the observed outcome be
\begin{equation*}
Y=g(v(X,D),\varepsilon )=\arg \max_{j=0,1,...,J}y_{j,D}^{\ast },
\end{equation*}%
where 
\begin{equation*}
y_{j,D}^{\ast }=v_{j}(X,D)+\varepsilon _{j} \mbox{ for } j =
1,2,...,J\text{; }y_{0,D}^{\ast }=0\text{.}
\end{equation*}%
In this case, the index $v\equiv (v_{j})_{j\leq J}$ and the errors $%
\varepsilon \equiv (\varepsilon _{j})_{j\leq J}$ are both $J$-dimensional.
The selection equation that determines $D$ is the same as (1.2). In this
case, we can replace $1\{Y\leq y\}$ by $1\{Y=y\}$ in the definition of $%
h_{1},h_{0},h_{1}^{\ast },h_{0}^{\ast }$ and $F_{g|u}(\cdot ;v)$. Then for $%
d=0,1$ and $j\leq J$,\ 
\begin{eqnarray*}
F_{g|u}(j;v(x,d)) &\equiv &E[1\{g(v(x,d),\varepsilon )=j\}|U=u] \\
&=&\Pr \left \{ v_{j}(x,d)+\varepsilon _{j}\geq v_{j^{\prime
}}(x,d)+\varepsilon _{j^{\prime }}\text{ }\forall j^{\prime }\leq J\mid
U=u\right \} \text{.}
\end{eqnarray*}%
By \citeasnoun{ruud2000} and \citeasnoun{ahnetal2018}, the mapping from $%
v\in \mathbb{R}^{J}$ to $(F_{g|u}(j;v):j\leq J)\in \mathbb{R}^{J}$ is smooth
and invertible provided that $\varepsilon \in \mathbb{R}^{J}$ has
non-negative density everywhere. This implies Assumption A-5.\bigskip

\noindent \textbf{Example 3}. \textbf{(Potential outcome from a Roy model)}
Consider a treatment effect model with an endogenous binary treatment $D$
and with the potential outcome determined by a latent Roy model. The Roy
model has also been studied extensively from both applied and theoretical
perspectives. See, for example, the literature survey in %
\citeasnoun{heckvyt-handbook} and the seminal paper in %
\citeasnoun{heckmanhonoreb}.

Here the observed outcome consists of two pieces:\ \ a continuous measure $%
Y=DY_{1}+(1-D)Y_{0}$ and a discrete indicator $W=DW_{1}+(1-D)W_{0}$ for $%
d=0,1$. These potential outcomes are given by 
\begin{equation*}
Y_{d}=\max_{j\in \{a,b\}}y_{j,d}^{\ast }\text{ and }W_{d}=\arg \max_{j\in
\{a,b\}}y_{j,d}^{\ast }
\end{equation*}%
where $a$ and $b$ index potential outcomes realized in different sectors,
with 
\begin{equation*}
y_{j,d}^{\ast }=v_{j}(X,d)+\varepsilon _{j} \mbox{ for } j\in \{a,b\}\text{.}
\end{equation*}%

The binary endogenous treatment $D$ is determined as in 
equation (1.2). For example, $D\in \{1,0\}$ indicates whether an individual
participates in a professional training program, $W_{d}\in \{a,b\}$
indicates the potential sector in which the individual is employed, $%
y^*_{j,d}$ is the potential wage from sector $j$ under treatment $D=d$, and $Y_{d}\in \mathbb{R}$ is the potential wage if the treatment status is $D=d$.\footnote{
	Note this differs from the classical approach that formulates treatment effects using Roy models (e.g. \citeasnoun{heckmanUrzuaVytlacil2006}) in that the potential outcome $Y_d$ itself is determined by a latent Roy model.}

As before, we maintain that $(X,Z)\bot (\varepsilon ,U)$.
The parameter of interest is 
\begin{equation*}
\Pr \{Y_{1}\leq y,W_{1}=a|X\}.
\end{equation*}%
By the independence condition that $(X,Z)\bot (\varepsilon ,U)$ and an application of the law of total probability, this conditional probability can be expressed in terms of directly identifiable quantities and the following counterfactual quantity 
\begin{eqnarray}
&&\Pr \{Y_{1}\leq y,W_{1}=a\mid X=x,Z=z,D=0\}  \notag \\
&=&\Pr \left \{ v_{b}(x,1)+\varepsilon _{b}<v_{a}(x,1)+\varepsilon _{a}\leq
y\mid U\geq P(z)\right \} \text{.}  \label{CF}
\end{eqnarray}%
Again, we seek to identify this counterfactual quantity by finding $\tilde{x}$ such that there exists $\tilde{z} $ with $ (\tilde{x},\tilde{z}) \in Supp(X,Z)$, $P(z) = P(\tilde{z}) $, and
\begin{equation}\label{fullEQ}
v_{a}(x,1)=v_{a}(\tilde{x},0)\text{ and }v_{b}(x,1)=v_{b}(\tilde{x},0)\text{.}  
\end{equation}%
This would allow us to recover the counterfactual conditional probability in  (\ref{CF}) as 
\begin{equation*}
\Pr \{Y_{0}\leq y,W_{0}=a\mid X=\tilde{x},Z=\tilde{z},D=0\} \text{.}
\end{equation*}

To find such a pair of $(x,\tilde{x})$, define $h_{d,W}(x,p,p^{\prime
}),h_{d,W}^{\ast }(x,p)$ by replacing $1\{Y\leq y\}$ with $1\{W=a\}$ in the
definition of $h_{d},h_{d}^{\ast }$ in Section \ref{Identification}.
Similarly, define $h_{d,Y}(x,y,p,p^{\prime }),h_{d,Y}^{\ast }(x,y,p)$ by
replacing $1\{Y\leq y\}$ with $1\{Y\leq y,W=a\}$ in the definition of $%
h_{d},h_{d}^{\ast }$ in Section \ref{Identification}. Then 
\begin{eqnarray*}
h_{d,W}(x,p_{1},p_{2}) &=&\int_{p_{2}}^{p_{1}}\Pr \{v_{b}(x,d)+\varepsilon
_{b}<v_{a}(x,d)+\varepsilon _{a}|U=u\}du\text{;} \\
h_{d,Y}(x,y,p_{1},p_{2}) &=&\int_{p_{2}}^{p_{1}}\Pr \{v_{b}(x,d)+\varepsilon
_{b}<v_{a}(x,d)+\varepsilon _{a}\leq y|U=u\}du\text{;}
\end{eqnarray*}%
and $h_{d,W}(x,p_{1},p_{2})$ and $h_{d,Y}(x,y,p_{1},p_{2})$ are both
identified over their respective domains by construction. 

Assume $(\varepsilon _{a},\varepsilon _{b})$ is continuously distributed with positive density over $\mathbb{R}^{2}$ conditional on all $u$, and the distribution of $ (\varepsilon_0,\varepsilon_1) $ given $U=u$ is continuous in $u $ over $[0,1]$. 
Then the statement 
\begin{eqnarray*}
\text{\textquotedblleft }h_{1,W}(x,p,p^{\prime })=h_{0,W}(\tilde{x}%
,p,p^{\prime })\text{, }h_{1,Y}(x,y,p,p^{\prime })=h_{0,Y}(%
\tilde{x},y,p,p^{\prime })\text{ } 
\text{for all }y\text{ and }p > p^{\prime }\text{ on }\mathcal{P}_{x}\cap \mathcal{P}_{\tilde{x}%
}\text{\textquotedblright }
\end{eqnarray*}%
holds true if and only if (\ref{fullEQ}) holds. 
To see this, first note that matching $%
h_{1,W}(x,p,p^{\prime })=h_{0,W}(\tilde{x},p,p^{\prime })$ requires 
\begin{equation}
v_{a}(x,1)-v_{b}(x,1)=v_{a}(\tilde{x},0)-v_{b}(\tilde{x},0)\text{,}
\label{CF1}
\end{equation}
while matching $h_{1,Y}(x,y,p,p^{\prime })=h_{0,Y}(\tilde{x},y,p,p^{\prime
}) $ at the same time requires 
\begin{equation}
v_{a}(x,1)=v_{a}(\tilde{x},0)\text{.}  \label{CF2}
\end{equation}%
Thus requiring (\ref{CF1}) and (\ref{CF2}) to hold jointly is equivalent to (\ref{fullEQ}).

It is worth mentioning that the identification strategy above also applies in a more general setup where the error term in the potential outcome is specific to the treatment, i.e., with $\varepsilon_j$ replaced by $\varepsilon_{j,d}$ in $y^*_{j,d}$. 
In such cases, the argument above remains valid under a ``rank similarity'' condition (that the marginal distribution of $\varepsilon_{j,d}$ is the same for $d\in\{0,1\}$).
The rank similarity condition has been used for identifying treatment effects in instrumental quantile regression models such as  \citeasnoun{chernozhukovhansenjoe}.

\setcounter{equation}{0}

\section{Extension}

\label{Extension}

The identification strategy we use requires finding matched pairs for $x$ in $\mathcal{X}^1$ and $\mathcal{X}^0$. 
In some cases, with the outcome being continuous, we can construct similar arguments for identifying a counterfactual quantity in a treatment effect model by matching different elements on the support of continuous outcomes.
To the best of our knowledge, this approach has not been explored in the literature on the effects of endogenous treatments. 
The following example illustrates this point. 

\noindent \textbf{Example 4.} \textbf{(Potential outcome with random
coefficients)} Random coefficient models are prominent in both the
theoretical and applied econometrics literature. They permit a flexible way to allow for conditional heteroscedasticity and unobserved heterogeneity.
For a survey and recent developments, see \citeasnoun{hsiaochapter}, \citeasnoun{hoderlein2010analyzing}, \citeasnoun{arellano2012identifying}, and \citeasnoun{masten2018random}. 

We consider a treatment effect model where the potential outcome is determined through random coefficients: 
\begin{equation*}
Y=DY_{1}+(1-D)Y_{0}\text{ where }Y_{d}=(\alpha _{d}+X^{\prime }\beta _{d})%
\text{ for }d=0,1
\end{equation*}%
and the binary endogenous treatment $D$ is determined as in equation (1.2). The \textit{random} intercepts $\alpha _{d}\in \mathbb{R}$
and the \textit{random} vectors of coefficients $\beta _{d}$ are given by 
\begin{equation*}
\alpha _{d}=\bar{\alpha}_{d}(X)+\eta _{d}\text{ and }\beta _{d}=\bar{\beta}%
_{d}(X)+\varepsilon _{d}
\end{equation*}%
where for any $x \in Supp(X)$ and $d \in \{0,1\}$, $(\bar{\alpha}_{d}(x),\bar{\beta}%
_{d}(x))\in \mathbb{R}^{K+1}$ is a vector of constant parameters while $\eta
_{d}\in \mathbb{R}$ and $\varepsilon _{d}\in \mathbb{R}^{K}$ are
unobservable noises.

As before, suppose some elements of $Z$ in the treatment equation are excluded from $X$. We allow the vector of unobservable terms $(\varepsilon_{1},\varepsilon _{0},\eta _{0},\eta _{1},U)$ to be arbitrarily correlated, and assume:%
\begin{equation}
\left( X,Z\right) \text{ }\bot \text{ }(\varepsilon _{1},\varepsilon _{0},\eta
_{0},\eta _{1},U)\text{.}  \label{indep}
\end{equation}%
Aslo, normalize the marginal distribution of $U$  to standard uniform, so that $\theta (Z)$ is directly identified as $P(Z)\equiv E(D|Z)$.

Our goal is to identify the distribution of potential outcomes $%
Y_{d}$ given $X=x$ for $d=0,1$. 
From this result, we can recover other parameters of interest such as average treatment effects, quantile treatment effects, etc.
As a preliminary step, we start by pinpointing a counterfactual item that is crucial for this identification question.

Let $G_{P|x}$ denote the distribution of $P\equiv P(Z)$ given $X=x$ (recall that its support is denoted as $\mathcal{P}_x$). 
This conditional distribution is directly identifiable from the data-generating process. 
By construction,%
\begin{equation*}
\Pr \{Y_{1}\leq y|X=x\}=\int \Pr \{Y_{1}\leq y|X=x,P=p\}dG_{P|x}(p)\text{,}
\end{equation*}%
where%
\begin{eqnarray} \label{dist_ltp}
&&\Pr \{Y_{1}\leq y|X=x,P=p\} \notag \\
&=&E\left[ D1\{Y_{1}\leq y\}|X=x,P=p\right] +E\left[ (1-D)1\{Y_{1}\leq
y\}|X=x,P=p\right] \text{.}
\end{eqnarray}%
The first term on the right-hand side of (\ref{dist_ltp}) is identified as%
\begin{equation*}
E[D1\{Y\leq y\}|X=x,P=p]\text{.}
\end{equation*}%
The second term on the right-hand side of (\ref{dist_ltp}), denoted by $\phi_0(x,y,p)$ , is counterfactual and can be written as%
\begin{eqnarray*}
\phi _{0}(x,y,p)&\equiv& E[1\{U\geq P\}1\{\alpha _{1}+X^{\prime }\beta
_{1}\leq y\}|X=x,P=p] \\
&=&E[1\{U\geq p\}1\{\bar{\alpha}_{1}(x)+\eta _{1}+x^{\prime }(\bar{\beta}%
_{1}(x)+\varepsilon _{1})\leq y\}] \\
&=&\int_{p}^{1}\Pr \{\eta _{1}+x^{\prime }\varepsilon _{1}\leq y-\bar{\alpha}%
_{1}(x)-x^{\prime }\bar{\beta}_{1}(x)|U=u\}du\text{.}
\end{eqnarray*}%
Hence identification of the conditional distribution of $Y_1$ amounts to identification of $\phi_0(\cdot) $.

As noted at the beginning of this section, we will identify the counterfactual $\phi_0(\cdot)$ by finding matched elements on the support of the observed outcome $Y$. 
This takes two steps.
First, we show that if for a given pair $(x,y)$ one can find $t(x,y)$ such that
\begin{equation}\label{defn_txy}
y-\bar{\alpha}_{1}(x)-x^{\prime }\bar{\beta}_{1}(x)=t(x,y)-\bar{\alpha}%
_{0}(x)-x^{\prime }\bar{\beta}_{0}(x)\text{,}
\end{equation}%
then one can use $t(x,y)$ to identify $\phi_0(x,y,p)$ for any $p\in \mathcal{P}_x$.
Specifically, for any $p$ on the support of $P$ given $X=x$, define%
\begin{eqnarray*}
h_{1}^{\ast }(x,y,p)&\equiv& E\left[ D1\left\{ Y\leq y\right\} |X=x,P=p%
\right]  \\
&=&E\left[ 1\left\{ U<P\right\} 1\left\{ \alpha _{1}+X^{\prime }\beta
_{1}\leq y\right\} |X=x,P=p\right] \\
&=&E\left[ 1\left\{ U<p\right\} 1\left\{
\alpha _{1}+x^{\prime }\beta _{1}\leq y\right\} \right]  \\
&=&\int_{0}^{p}\Pr \{\eta _{1}+x^{\prime }\varepsilon _{1}\leq y-\bar{\alpha}%
_{1}(x)-x^{\prime }\bar{\beta}_{1}(x)|U=u\}du,
\end{eqnarray*}%
where the second equality uses (\ref{indep}). Likewise, under (\ref{indep})
we have: 
\begin{eqnarray*}
h_{0}^{\ast }(x,y,p) &\equiv& E\left[ (1-D)1\left\{ Y\leq y\right\} |X=x,P=p%
\right]  \\
&=&\int_{p}^{1}\Pr \{\eta _{0}+x^{\prime }\varepsilon _{0}\leq y-\bar{\alpha}%
_{0}(x)-x^{\prime }\bar{\beta}_{0}(x)|U=u\}du\text{.}
\end{eqnarray*}%
Assume\footnote{%
Such a distributional equality condition has been used to motivate the \emph{rank similarity}
condition imposed frequently in the econometrics literature -- see, for example, %
\citeasnoun{chernozhukovhansen}, \citeasnoun{vytlacilyildiz}, \citeasnoun{chenstackhan}, \citeasnoun{franlef}, \citeasnoun{dongshen}.}
\begin{equation}
F_{\left( \eta _{1},\varepsilon _{1}\right) |U=u}=F_{\left( \eta _{0},\varepsilon
_{0}\right) |U=u}\text{ for all }u\in \lbrack 0,1]\text{.}  \label{eqCF}
\end{equation}%
Under (\ref{eqCF}), we have%
\begin{equation}
\phi _{0}(x,y,p)=\int_{p}^{1}\Pr \{ \eta _{0}+x^{\prime }\varepsilon _{0}\leq y-%
\bar{\alpha}_{1}(x)-x^{\prime }\bar{\beta}_{1}(x)|U=u\}du\text{.}
\label{swapEps}
\end{equation}%
Then by definition of $t(x,y)$ in (\ref{defn_txy}),%
\begin{eqnarray*}
h_{0}^{\ast }(x,t(x,y),p) &\equiv &\int_{p}^{1}\Pr \{ \eta _{0}+x^{\prime
}\varepsilon _{0}\leq t(x,y)-\bar{\alpha}_{0}(x)-x^{\prime }\bar{\beta}%
_{0}(x)|U=u\}du \\
&=&\int_{p}^{1}\Pr \{ \eta _{0}+x^{\prime }\varepsilon _{0}\leq y-\bar{\alpha}%
_{1}(x)-x^{\prime }\bar{\beta}_{1}(x)|U=u\}du \\
&=&\phi _{0}(x,y,p).
\end{eqnarray*}%
Thus the counterfactual $\phi _{0}(x,y,p)$ would be identified as $h_{0}^{\ast }(x,t(x,y),p)$.

The second step is to show that for each pair $(x,y)$ we can indeed uniquely recover $t(x,y) $ using quantities that are identifiable in the data-generating process. 
To do so, we define two auxiliary functions as follows: for $p_{1}>p_{2}$ on the support of $P$ given $X=x$, let 
\begin{eqnarray*}
h_{1}(x,y,p_{1},p_{2}) &\equiv &h_{1}^{\ast }(x,y,p_{1})-h_{1}^{\ast
}(x,y,p_{2}) \\
&=&\int_{p_{2}}^{p_{1}}\Pr \{ \eta _{1}+x^{\prime }\varepsilon _{1}<y-\bar{%
\alpha}_{1}(x)-x^{\prime }\bar{\beta}_{1}(x)|U=u\}du\text{;}
\end{eqnarray*}%
and%
\begin{eqnarray*}
h_{0}(x,y,p_{1},p_{2}) &\equiv &h_{0}^{\ast }(x,y,p_{2})-h_{0}^{\ast
}(x,y,p_{1}) \\
&=&\int_{p_{2}}^{p_{1}}\Pr \{ \eta _{0}+x^{\prime }\varepsilon _{0}<y-\bar{%
\alpha}_{0}(x)-x^{\prime }\bar{\beta}_{0}(x)|U=u\}du\text{.}
\end{eqnarray*}%
Suppose $\eta _{d}+x^{\prime }\varepsilon _{d}$ is continuously distributed
over $\mathbb{R}$ for all values of $x$ conditional on all $u \in [0,1]$.
Then for any fixed pair $(x,y)$ and $p_{1}>p_{2}$,%
\begin{equation*}
h_{1}(x,y,p_{1},p_{2})=h_{0}(x,t(x,y),p_{1},p_{2})
\end{equation*}%
if and only if%
\begin{equation*}
t(x,y)=y-\bar{\alpha}_{1}(x)-x^{\prime }\bar{\beta}_{1}(x)+\bar{\alpha}%
_{0}(x)+x^{\prime }\bar{\beta}_{0}(x)\text{.}
\end{equation*}%

\setcounter{equation}{0}

\section{Estimation}

\label{estimation}

In this section, we outline estimation procedures from a random sample of the observed variables that are motivated by our identification results. 
We first describe an estimation procedure for the parameter $E[Y_1]$ in the first three examples. 
Recall $\mathcal{P}_{x}$ denotes the support of $P(Z) \equiv P $ given $X=x$. 
Let $f_P(.|x)$ denote the density of $P(Z)$ given $X=x$, and define
\begin{equation*}
\mathcal{P}_{x}^{c} \equiv \left \{ p\text{: }f_{P}(p|x)>c\right \} \mbox{ for a known } c>0.
\end{equation*}%
For simplicity, assume 
\begin{equation*}
1-c_{0}>P(Z)>c_{0} \mbox{ for a known } c_0 > 0 \mbox{ almost surely}.
\end{equation*}%
Define a measure of distance between $h_{1}(x_{1},\cdot )$ and $%
h_{0}(x_{0},\cdot )$ as follows:
\begin{eqnarray*}
&&\left \Vert h_{1}(x_{1},\cdot )-h_{0}(x_{0},\cdot )\right \Vert \\
&=&\left \{ \int \int \int \left( \int_{p_{2}}^{p_{1}} (F_{g|u}(y;v(x_1,1))
- F_{g|u}(y;v(x_0,0)) du \right)^{2} I \left( p_{1},p_{2}\in
\mathcal{P}_{x}^{c}\right) w(y)dydp_{1}dp_{2}\right \}^{1/2}
\end{eqnarray*}
where $w(y)$ is a chosen weight function. 

Consider the case when $h_{0}(x,y,p_{1},p_{2})$, $h_{1}(x,y,p_{1},p_{2})$ and $P(z)$ are known. For any given $x_{i}$, let $\tilde x_{i}$ be such that%
\begin{equation*}
\left \Vert h_{0}(\tilde x_{i},\cdot ) - h_{1}(x_{i},\cdot ) \right \Vert =0,
\end{equation*}
which, under Assumption A-5 in Section \ref{Identification}, is equivalent to 
\begin{equation*}
v(\tilde x_{i},0) = v(x_{i},1).
\end{equation*}%
Let $P_i$ be shorthand for $P(Z_i)$ and define
\begin{equation*}
Y_{i}^{*} \equiv E[Y|D=0, \| h_0( X,\cdot) - h_1(X_i,\cdot) \|=0,P=P_i].
\end{equation*}%
This conditional expectation equals $E[ Y | D=0, v( X , 0 ) = v( X_i , 1), P=P_i] $, which in turn equals $E( Y_1 | D=0, X=X_i, P=P_i) $. 

The parameter of interest $\Delta \equiv E[Y_{1}]$ can be written as $ \Delta = E[D_iY_i+(1-D_i)Y^*_i]$.
Therefore, we estimate $\Delta$ by its sample analog, after replacing $Y^*_i$ with its Nadaraya-Watson estimates. 
That is, we estimate $\Delta$ by
\begin{equation*}
\hat{\Delta}=\frac{1}{n}\sum_{i=1}^{n}\left( D_{i}Y_{i}+(1-D_{i})\hat{Y}%
_{i}\right),
\end{equation*}%
where $\hat{Y}_i$ is a kernel regression estimator of $Y^*_i$, using knowledge of $h_0,h_1$ and $P(Z)$.  
Likewise, for estimating the conditional mean $E(Y_1|X \in A)$ where $A$ is a generic subset of the support of $X$, we use a weighted version%
\begin{equation*}
\hat{\Delta}_{A}=\frac{\frac{1}{n}\sum_{i=1}^{n}1\left \{ X_{i}\in A\right
\} \left( D_{i}Y_{i}+(1-D_{i})\hat{Y}_{i}\right) }{\frac{1}{n}%
\sum_{i=1}^{n}1\left \{ X_{i}\in A\right \} }.
\end{equation*}%
Limiting distribution theory for each of these estimators follows from
identical arguments in \citeasnoun{vytlacilyildiz}. Here we formally state
the theorem for the first estimator:

\begin{theorem}
\label{theorem1} Suppose Assumptions (A-1) to (A-6) hold, and $Y_1$ has positive and finite second moments. Then
\begin{equation*}
\sqrt{n}(\hat \Delta-\Delta)\overset{d}{\rightarrow} \mathbb{N}(0,\Sigma),
\end{equation*}
where 
\begin{equation*}
\Sigma=Var(E[Y_1|X,P,D])+E[PVar(Y_1|X,P,D=1)].
\end{equation*}
\end{theorem}

Next, we describe an estimation procedure for the distributional treatment effect in Example 4, where potential outcomes depend on random coefficients. 
In this case, the parameter of interest is, for a given value $y \in \mathbb{R}$, 
\begin{equation*}
\Delta_2(y)=\Pr \{Y_{1}\leq y\} \text{.}
\end{equation*}%
For fixed values of $y$ and $p_{1}>p_{2}$, we propose to estimate $t(x,y)$ as 
\begin{equation*}
\hat{t}(x,y,p_{1},p_{2})=\arg
\min_{t}(h_{1}(x,y,p_{1},p_{2})-h_{0}(x,t,p_{1},p_{2}))^{2},
\end{equation*}%
and then average over values of $p_{1},p_{2}$: 
\begin{equation*}
\hat{\tau}(x,y)=\frac{1}{n(n-1)}\sum_{i\neq j}I[P_{i}>P_{j}]\hat{t}%
(x,y,P_{i},P_{j}).
\end{equation*}%
An infeasible estimator for $\Delta_2(y)$, which assumes $%
t(x,y)$ is known, would be%
\begin{equation*}
\hat{\Delta}_2(y)=\frac{1}{n}\sum_{i=1}^{n} \left( D_{i}1\{Y_{i}\leq
y\}+(1-D_{i})1\{Y_{i}\leq t(X_{i},y)\} \right) \text{.}
\end{equation*}%
In practice, for feasible estimation, one needs to replace $t(x,y)$ with its estimator $\hat{\tau}(x,y)$.


We conclude this section with some discussion about the computational aspects.
The computational costs for finding the suitable pairs $(x,\tilde{x})$ are modest in comparison with typical semi- or nonparametric methods in the literature. 
To see this, consider the case with $J=2$, where $v=(v_{1},v_{2})\in \mathbb{R}^{2}$ consists of two index functions $v_{j}(x,d):\mathbb{R}^{K}\times \{0,1\}\rightarrow R$ for $j=1,2$. 
The actual dimensions that matter in implementation are: (i) $K$ in the estimation of $h_{1}$ and $h_{0}$ in the first stage, and (ii)  the dimension of the indexes to be matched in the second stage: $J=2$. 
Therefore, the dimensionality that causes  difficulty in estimation is $\max \left\{J,K\right\} $, not $J\times K$. 
In our view, this is no more difficult than implementing estimators that are based on matching or pairwise comparison, such as Blundell and Powell (2003).

\setcounter{equation}{0}

\section{Simulation Study}

\label{simulations}

This section presents simulation evidence for the performance of the
estimator in Section \ref{estimation}, for both the Average Treatment Effect and the Distributional Treatment Effect.

\subsection{Designs with continuous outcomes}

We report results for both our estimator and that in \citeasnoun{vytlacilyildiz}, for several designs where the potential outcomes are real-valued and continuous. 
These include designs where the monotonicity condition fails, and designs where the disturbance terms in the outcome equation are multi-dimensional.

Throughout all designs, we model the treatment or dummy endogenous variable
as 
\begin{equation*}
D=I[Z-U>0],
\end{equation*}
where $Z,U$ are independent standard normal. We experiment with the
following designs for the outcome.
\begin{description}

\item[(Design 1)] $Y=X+0.5\cdot D +\varepsilon$, where $X$ is standard normal while $(\varepsilon, U)$ are bivariate normal, with mean 0, variance 1, and correlation $\rho_v \in \{0, 0.25, 0.5\}$.


\item[(Design 2)] $ Y=X+0.5\cdot D+(X+D)\cdot \varepsilon $, where $X$ is standard normal while $(\varepsilon, U)$ are bivariate normal, with mean 0, variance 1, and correlation $\rho_v \in \{0, 0.25, 0.5\}$.

\item[(Design 3)] $ Y=(X+0.5 \cdot D+\varepsilon)^2$, where $X$ is standard normal while $(\varepsilon, U)$ are bivariate normal, each with mean 0, variance 1, and correlation $\rho_v \in \{0, 0.25, 0.5\}$

\end{description}

We note that the monotonicity condition holds in Design 1 but fails in the other two designs. For each of these designs, we report results for estimating $E[Y_1]$, i.e., the mean potential outcome under treatment $D=1$. 
The two estimators reported in the simulation study are our estimator proposed in Section \ref{estimation} and the one proposed in \citeasnoun{vytlacilyildiz}. 
The summary statistics, scaled by the true parameter value, Mean Bias, Median Bias, Root Mean Squared Error, (RMSE), and Median Absolute Deviation (MAD) are evaluated for sample sizes of $n = 100, 200, 400$ for 401 replications. 

Results for each of these designs are reported in Tables 1 to 3 respectively. In implementing our estimator, we assume the propensity score function is known, and conduct next stage estimation using a nonparametric kernel estimator with normal kernel function, and a bandwidth of $n^{-1/5}$. This rate reflects ``undersmoothing" as there are two regressors, the propensity score and the regressor $X$. For the estimator in \citeasnoun{vytlacilyildiz}, which involves the derivative of conditional expectation functions as well,  we also report results for an infeasible version of their estimator, assuming such functions, as well as the propensity scores, are known.
To implement the second stage of our estimator, in calculating the distance $\|h_1(x_i,\cdot)-h_0(x_0,\cdot)\|$ we used an evenly spaced grid of values for $y$, and selected $n/50$ grid points, with $n $ denoting the sample size. 

The results indicate the desirable properties of our estimator, generally agreeing with Theorem \ref{theorem1}. In all designs, our estimator has small values for bias and RMSE, with the value of RMSE decreasing as the sample size grows. 
In contrast, the procedure based on \citeasnoun{vytlacilyildiz} only performs well in Design 1, with the sizes of bias and RMSE comparable to those using our method. 
As in our estimator, these values decrease as the sample size grows, which is expected, as the monotonicity condition they require is satisfied in this design. 
In this case, their approach has smaller standard errors largely due to the relatively simpler structure of the infeasible version of the estimator, but their biases persist even when the sample size increases.

For Designs 2 and 3, where the monotonicity condition is violated, the estimator proposed in \citeasnoun{vytlacilyildiz} does not perform well. Table 2 shows that in Design 2 both the bias and RMSE of their estimator are generally decreasing slowly with the sample size.
Results for their estimator are better in Design 3 in Table 3, but the bias hardly converges with the sample size, and is much larger compared to our estimator.


\begin{center}

\scalefont{0.75}{

Table 1

\begin{tabular}{||l|ccc|ccc||}
\hline
\  & \  & CKT &  &  & VY &  \\ \hline\hline
$\rho _{v}$ & $0$ & $1/4$ & $1/2$ & $0$ & $1/4$ & $1/2$ \\ \hline
n=100 & \  & \  & \  & \  & \  &  \\ 
MEAN BIAS & -0.0170 & 0.0229 & -0.0435 & -0.1302 & -0.1676 & -0.2018 \\ 
MEDIAN BIAS & -0.0137 & 0.0124 & -0.0653 & -0.1318 & -0.1678 & -0.2087 \\ 
RMSE & 0.4936 & 0.4800 & 0.4945 & 0.3308 & 0.3337 & 0.3546 \\ 
MAD & 0.3289 & 0.3328 & 0.3156 & 0.2200 & 0.2271 & 0.2546 \\ 
n=200 & \  & \  & \  & \  & \  &  \\ 
MEAN BIAS & 0.0032 & -0.0024 & -0.0069 & -0.0864 & -0.1299 & -0.1766 \\ 
MEDIAN BIAS & -0.0102 & -0.0141 & -0.0314 & -0.0934 & -0.1277 & -0.1679 \\ 
RMSE & 0.3355 & 0.3367 & 0.3521 & 0.2293 & 0.2457 & 0.2711 \\ 
MAD & 0.2240 & 0.2228 & 0.2517 & 0.1594 & 0.1676 & 0.1865 \\ 
n=400 & \  & \  & \  & \  & \  &  \\ 
MEAN BIAS & -0.0187 & 0.0101 & -0.0055 & -0.0584 & -0.11134 & -0.1593 \\ 
MEDIAN BIAS & -0.0261 & 0.0128 & -0.0065 & -0.0592 & -0.1162 & -0.1572 \\ 
RMSE & 0.2496 & 0.2489 & 0.2578 & 0.2049 & 0.1867 & 0.2167 \\ 
MAD & 0.1523 & 0.1732 & 0.1659 & 0.1197 & 0.1345 & 0.1605 \\ \hline\hline
\end{tabular}

\bigskip \bigskip 

Table 2

\begin{tabular}{||l|ccc|ccc||}
\hline
\  & \  & CKT &  &  & VY &  \\ \hline\hline
$\rho _{v}$ & $0$ & $1/4$ & $1/2$ & $0$ & $1/4$ & $1/2$ \\ \hline
n=100 & \  & \  & \  & \  & \  &  \\ 
MEAN BIAS & 0.0109 & 0.0397 & -0.0671 & -0.1509 & -0.2875 & -0.4207 \\ 
MEDIAN BIAS & 0.0151 & 0.0227 & -0.0939 & -0.1590 & -0.2918 & -0.4262 \\ 
RMSE & 0.5089 & 0.2737 & 0.4853 & 0.3524 & 0.4199 & 0.5289 \\ 
MAD & 0.3395 & 0.2447 & 0.3105 & 0.2419 & 0.30898 & 0.4310 \\ 
n=200 & \  & \  & \  & \  & \  &  \\ 
MEAN BIAS & 0.0322 & 0.0143 & -0.0311 & -0.1273 & -0.2559 & -0.3875 \\ 
MEDIAN BIAS & 0.0159 & 0.0054 & -0.0543 & -0.1310 & -0.2553 & -0.3884 \\ 
RMSE & 0.3487 & 0.3444 & 0.3455 & 0.2622 & 0.3407 & 0.4475 \\ 
MAD & 0.2317 & 0.2297 & 0.2552 & 0.1782 & 0.2624 & 0.3884 \\ 
n=400 & \  & \  & \  & \  & \  &  \\ 
MEAN BIAS & 0.0088 & 0.0269 & -0.0294 & -0.0962 & -0.2247 & -0.3708 \\ 
MEDIAN BIAS & 0.0007 & 0.0244 & -0.0309 & -0.0982 & -0.2255 & -0.3769 \\ 
RMSE & 0.2578 & 0.2557 & 0.2549 & 0.1920 & 0.2764 & 0.4037 \\ 
MAD & 0.1649 & 0.1733 & 0.1606 & 0.1354 & 0.2283 & 0.3769 \\ \hline\hline
\end{tabular}

\pagebreak 

Table 3

\begin{tabular}{||l|ccc|ccc||}
\hline
\  & \  & CKT &  &  & VY &  \\ \hline\hline
$\rho _{v}$ & $0$ & $1/4$ & $1/2$ & $0$ & $1/4$ & $1/2$ \\ \hline
n=100 & \  & \  & \  & \  & \  &  \\ 
MEAN BIAS & -0.0097 & -0.0070 & 0.0019 & -0.0691 & -0.0898 & -0.1066 \\ 
MEDIAN BIAS & -0.0233 & -0.0101 & -0.0240 & -0.0799 & -0.0925 & -0.1178 \\ 
RMSE & 0.1893 & 0.2085 & 0.2126 & 0.1546 & 0.1630 & 0.1701 \\ 
MAD & 0.1398 & 0.1342 & 0.1374 & 0.1125 & 0.1178 & 0.1315 \\ 
n=200 & \  & \  & \  & \  & \  &  \\ 
MEAN BIAS & -0.0108 & -0.0069 & -0.0068 & -0.0609 & -0.0765 & -0.0968 \\ 
MEDIAN BIAS & -0.0148 & -0.0033 & -0.0099 & -0.0674 & -0.0769 & -0.1017 \\ 
RMSE & 0.1372 & 0.1434 & 0.1424 & 0.1163 & 0.1262 & 0.1369 \\ 
MAD & 0.0949 & 0.0989 & 0.0953 & 0.0855 & 0.0887 & 0.1078 \\ 
n=400 & \  & \  & \  & \  & \  &  \\ 
MEAN BIAS & -0.0073 & -0.0014 & -0.0026 & -0.0583 & -0.0725 & -0.0889 \\ 
MEDIAN BIAS & -0.0149 & -0.0023 & -0.0029 & -0.0610 & -0.0751 & -0.0887 \\ 
RMSE & 0.1084 & 0.0994 & 0.0989 & 0.0924 & 0.1007 & 0.1131 \\ 
MAD & 0.0697 & 0.0685 & 0.0654 & 0.0689 & 0.0788 & 0.0901 \\ \hline\hline
\end{tabular}

}

\end{center}


We also report estimator performance in samples simulated from a model where potential outcomes are determined by random coefficients and dummy endogenous variables. 
It is important to note that for this design, the estimator in \citeasnoun{vytlacilyildiz} does not apply. 
This is because different values of $x$ lead to different distributions of the composite error $\eta_d + x^{\prime }\epsilon_d$. 
Our contribution in Section \ref{Extension} is to propose a new approach based on matching different values of the observed outcome $y$, rather than the exogenous covariates $x$. 
Based on the counterfactual framework discussed in Section \ref{Extension}, here the treatment variable $D$ is modeled as the same way as in the first three designs, with the regressor $X$ being standard normal. 
For both $Y_0,Y_1$, the intercepts were modeled as constants (0 and 1, respectively) and the additive error terms were each standard normal. 
For the random slopes, the means were 1 and 2 respectively, and the additive error terms were also standard normal, independent of all other disturbance terms and each other.
Here we use the procedure in Section \ref{Extension} to estimate the parameter $\Delta_2=P(Y_1<y)$, where in the simulation, we set $y=1$.  

Results for this design with random coefficients are reported in Table 4.
The same four summary statistics are reported for sample sizes $n \in \{100,200,400\}$, based on 401 replications.
The estimator proposed in Section \ref{estimation} performs well.
The bias and RMSE are much smaller for a bigger sample with $n=400$ than for smaller samples with $n = 100$ and $n=200$, indicating convergence at the parametric rate. 


\pagebreak 

\begin{center}

{\scalefont{0.75}

Table 4

\begin{tabular}{||l|ccc||}
\hline
\  & CKT &  &  \\ \hline\hline
$\rho _{v}$ & $0$ & $1/4$ & $1/2$ \\ \hline
n=100 & \  & \  &  \\ 
MEAN BIAS & 0.0109 & -0.0086 & 0.0038 \\ 
MEDIAN BIAS & 0.0000 & -0.0064 & 0.0126 \\ 
RMSE & 0.1011 & 0.0979 & 0.0955 \\ 
MAD & 0.0600 & 0.0648 & 0.0652 \\ 
n=200 & \  & \  & \  \\ 
MEAN BIAS & -0.0050 & -0.0150 & 0.0095 \\ 
MEDIAN BIAS & -0.0100 & -0.0161 & 0.0029 \\ 
RMSE & 0.0669 & 0.0669 & 0.0665 \\ 
MAD & 0.0400 & 0.0454 & 0.0457 \\ 
n=400 & \  & \  & \  \\ 
MEAN BIAS & 0.0012 & -0.0132 & 0.0074 \\ 
MEDIAN BIAS & 0.0049 & -0.0162 & 0.0077 \\ 
RMSE & 0.0501 & 0.0494 & 0.0495 \\ 
MAD & 0.0349 & 0.0325 & 0.0360 \\ \hline\hline
\end{tabular}
}

\end{center}

\subsection{Designs with discrete potential outcomes}

We also report the performance of our estimator in a sample drawn from
Example 2, where the observed outcomes are determined in a multinomial
choice model:%
\[
Y_{i}(D_{i})=\arg \max_{j\in \{0,1,2\}}Y_{i,j}^{\ast }(D_{i})\text{,}
\]%
where the potential outcomes are 
$Y_{i,j}^{\ast
}(d)\equiv v_{j}(X_{i,j},d)+\varepsilon _{i,j}$ with%
\[
v_{j}(X_{i,j},d)=\alpha _{j}(d)+X_{i,j}\beta _{j}(d)
\]
for $j=1,2$, and $v_0(X_{i,j},d) =0$ by way of normalization.  
The exogenous covariates $X_{i,j}$ are drawn independently from standard normal, and the intercepts and slope coefficients are:%
\begin{eqnarray*}
\alpha _{1}(0)=0\text{; }\alpha _{2}(0)=1\text{; }\alpha _{1}(1)=1\text{; 
}\alpha _{2}(1)=2\text{; } 
\beta _{1}(0)=0.8\text{; }\beta _{2}(0)=1\text{; }\beta _{1}(1)=1%
\text{; }\beta _{2}(1)=2\text{.}
\end{eqnarray*}%
For each individual $i$, the binary treatment $D_{i}$ is determined as before:
\[
D_{i}=1\{U_{i} < Z_{i}\}\text{,}
\]%
where the instrument $Z_{i}$ is independently drawn from a standard uniform distribution. 
The marginal distribution of the selection error $U_{i}$ is standard uniform. 
Conditional on $U_{i}=u$, the outcome errors $\varepsilon_{i,j}$ are independent across $j=0,1,2$ and are distributed as type-1 extreme value with unit variance and means $(0,\delta u,2\delta u)$ respectively, where $\delta$ is a parameter to be specified in the data-generating process. 
We adopt this specification as it allows for substantial dependence between $U_{i}$ and $\varepsilon _{i}\equiv (\varepsilon_{i,j})_{j=0,1,2}$.

Our simulation study uses sample sizes $n\in \{250,500,1000,2000\}$. 
For each sample size $n$, we generate $S=400$ independent samples from the data-generating process above. 
Throughout this section, we focus on estimating a conditional distribution of \textit{potential} outcomes $\Pr\{Y_d=j|X \in \omega\}$ for $d\in \{0,1\}$ and $ j \in \{1,2\}$, where $\omega \equiv\{x:x_{j}\in \lbrack -1,1]$ for $j=1,2\}$ is a subset of the support of covariates. 

As a benchmark, we first implement the \textit{infeasible} version of the estimator in Section \ref{estimation}, where knowledge of $h_{1}^{\ast }(\cdot )$ and $h_{0}^{\ast }(\cdot )$ are used for finding $(x,\tilde{x})$ with $v(x,1)=v(\tilde{x},0)$ for estimating $\Pr\{Y_1=j|X \in \omega\}$ (or $v(x,0)=v(\tilde{x},1)$ for estimating $\Pr\{Y_0=j|X \in \omega\}$). 
In this case, $F_{g|u}$ has a known close form, and we use numerical integration via mid-point approximation to calculate $h_{1}^{\ast }$ and $h_{0}^{\ast }$. 
Table 5 shows the mean bias and mean squared error (M.S.E.) for the infeasible estimator calculated from the $S$ simulated samples. 
We report these measures for $Pr\{Y_d = j | X \in \omega \}$ for $d=0,1$ and $j=1,2$. 
In the last two columns of the table, we report the M.S.E. for the \textit{full} vector summarizing the probability mass function of $Y_d$, i.e., $[\Pr\{Y_d=1|X \in \omega \}, \Pr\{Y_d=2|X \in \omega \}]$. 

The M.S.E. in Table 5 diminishes at a root-n rate that is proportional to the sample size.
While in most cases the mean bias decreases with the sample size, the root-n rate of convergence appears to be substantially driven by the diminishing variance of the estimator.
The performance does not vary substantively across different designs with the parameter values that affect the strength of correlation between the structural error $\varepsilon _{i,j}$ and the selection error $U_{i}$, i.e., the parameter values $\delta=(1/4,1/3,1/2)$.

We then construct a feasible version of the estimator by using $h_{1}^{\ast }$ and $h_{0}^{\ast }$ with their corresponding kernel estimates.
Specifically, we use bivariate Gaussian kernels with bandwidths $1.06\hat{\sigma}^{-1/7}$, where $\hat{\sigma}$ denotes the sample standard deviation of components in $(X_{i},P_{i})$. 
Table 6 reports the mean bias and M.S.E. of this feasible estimator in the same data-generating process as in Table 5. 
The estimation errors in Table 6 are overall larger than those reported in Table 5 but demonstrate similar patterns of convergence in most cases, even though the convergence appears to be slower for the MSE of $F_{Y_{0}|X\in \omega}$ and $F_{Y_{1}|X\in \omega}$ when $\delta $ is large. 
The difference in the magnitude of estimation errors across Table 5
and Table 6 is attributable to the estimation error in $h_{1}^{\ast }$ and $h_{2}^{\ast }$. All in all, we conclude our estimator has decent
finite-sample performance in these designs.

\pagebreak 

{\scalefont{0.75}
\begin{center}

Table 5. Performance of Infeasible Estimator

\begin{tabular}{cc|cc|cc|cc|cc|c|c}
\hline
&  & \multicolumn{2}{|c|}{$\Pr\{Y_0 = 1|\omega\}$} & \multicolumn{2}{|c|}{$\Pr\{Y_0 = 2|\omega\}$}
& \multicolumn{2}{|c|}{$\Pr\{Y_1 = 1|\omega\}$} & \multicolumn{2}{|c|}{$\Pr\{Y_1 = 2|\omega\}$} & 
$F[Y_0|\omega]$ & $F[Y_1|\omega]$ \\ \hline\hline
$\delta $ & $n$ & Bias & MSE & Bias & MSE & Bias & MSE & Bias & MSE & MSE & 
MSE \\ \hline
\multicolumn{1}{r}{1/4} & \multicolumn{1}{r|}{250} & \multicolumn{1}{|r}{
0.01305} & \multicolumn{1}{r|}{0.00225} & \multicolumn{1}{|r}{-0.00120} & 
\multicolumn{1}{r|}{0.00328} & \multicolumn{1}{|r}{-0.00753} & 
\multicolumn{1}{r|}{0.00254} & \multicolumn{1}{|r}{0.00276} & 
\multicolumn{1}{r|}{0.00275} & \multicolumn{1}{|r|}{0.00554} & 
\multicolumn{1}{|r}{0.00254} \\ 
\multicolumn{1}{r}{1/4} & \multicolumn{1}{r|}{500} & \multicolumn{1}{|r}{
0.00970} & \multicolumn{1}{r|}{0.00113} & \multicolumn{1}{|r}{-0.00215} & 
\multicolumn{1}{r|}{0.00175} & \multicolumn{1}{|r}{-0.00693} & 
\multicolumn{1}{r|}{0.00121} & \multicolumn{1}{|r}{0.00227} & 
\multicolumn{1}{r|}{0.00144} & \multicolumn{1}{|r|}{0.00288} & 
\multicolumn{1}{|r}{0.00121} \\ 
\multicolumn{1}{r}{1/4} & \multicolumn{1}{r|}{1000} & \multicolumn{1}{|r}{
0.00910} & \multicolumn{1}{r|}{0.00077} & \multicolumn{1}{|r}{-0.00269} & 
\multicolumn{1}{r|}{0.00095} & \multicolumn{1}{|r}{-0.00389} & 
\multicolumn{1}{r|}{0.00078} & \multicolumn{1}{|r}{0.00146} & 
\multicolumn{1}{r|}{0.00080} & \multicolumn{1}{|r|}{0.00172} & 
\multicolumn{1}{|r}{0.00078} \\ 
\multicolumn{1}{r}{1/4} & \multicolumn{1}{r|}{2000} & \multicolumn{1}{|r}{
0.00871} & \multicolumn{1}{r|}{0.00041} & \multicolumn{1}{|r}{-0.00447} & 
\multicolumn{1}{r|}{0.00051} & \multicolumn{1}{|r}{-0.00376} & 
\multicolumn{1}{r|}{0.00038} & \multicolumn{1}{|r}{0.00110} & 
\multicolumn{1}{r|}{0.00043} & \multicolumn{1}{|r|}{0.00092} & 
\multicolumn{1}{|r}{0.00038} \\ \hline
\multicolumn{1}{r}{1/3} & \multicolumn{1}{r|}{250} & \multicolumn{1}{|r}{
0.01825} & \multicolumn{1}{r|}{0.00269} & \multicolumn{1}{|r}{-0.00955} & 
\multicolumn{1}{r|}{0.00339} & \multicolumn{1}{|r}{-0.00350} & 
\multicolumn{1}{r|}{0.00248} & \multicolumn{1}{|r}{-0.00262} & 
\multicolumn{1}{r|}{0.00276} & \multicolumn{1}{|r|}{0.00608} & 
\multicolumn{1}{|r}{0.00275} \\ 
\multicolumn{1}{r}{1/3} & \multicolumn{1}{r|}{500} & \multicolumn{1}{|r}{
0.00696} & \multicolumn{1}{r|}{0.00126} & \multicolumn{1}{|r}{-0.00144} & 
\multicolumn{1}{r|}{0.00166} & \multicolumn{1}{|r}{-0.00656} & 
\multicolumn{1}{r|}{0.00123} & \multicolumn{1}{|r}{0.00330} & 
\multicolumn{1}{r|}{0.00130} & \multicolumn{1}{|r|}{0.00292} & 
\multicolumn{1}{|r}{0.00144} \\ 
\multicolumn{1}{r}{1/3} & \multicolumn{1}{r|}{1000} & \multicolumn{1}{|r}{
0.01170} & \multicolumn{1}{r|}{0.00079} & \multicolumn{1}{|r}{-0.00842} & 
\multicolumn{1}{r|}{0.00097} & \multicolumn{1}{|r}{-0.00315} & 
\multicolumn{1}{r|}{0.00070} & \multicolumn{1}{|r}{0.00007} & 
\multicolumn{1}{r|}{0.00081} & \multicolumn{1}{|r|}{0.00175} & 
\multicolumn{1}{|r}{0.00080} \\ 
\multicolumn{1}{r}{1/3} & \multicolumn{1}{r|}{2000} & \multicolumn{1}{|r}{
0.00918} & \multicolumn{1}{r|}{0.00045} & \multicolumn{1}{|r}{-0.00625} & 
\multicolumn{1}{r|}{0.00060} & \multicolumn{1}{|r}{-0.00486} & 
\multicolumn{1}{r|}{0.00043} & \multicolumn{1}{|r}{0.00273} & 
\multicolumn{1}{r|}{0.00044} & \multicolumn{1}{|r|}{0.00105} & 
\multicolumn{1}{|r}{0.00043} \\ \hline
\multicolumn{1}{r}{1/2} & \multicolumn{1}{r|}{250} & \multicolumn{1}{|r}{
0.01322} & \multicolumn{1}{r|}{0.00235} & \multicolumn{1}{|r}{-0.00770} & 
\multicolumn{1}{r|}{0.00323} & \multicolumn{1}{|r}{-0.00359} & 
\multicolumn{1}{r|}{0.00253} & \multicolumn{1}{|r}{-0.00158} & 
\multicolumn{1}{r|}{0.00295} & \multicolumn{1}{|r|}{0.00559} & 
\multicolumn{1}{|r}{0.00248} \\ 
\multicolumn{1}{r}{1/2} & \multicolumn{1}{r|}{500} & \multicolumn{1}{|r}{
0.01426} & \multicolumn{1}{r|}{0.00130} & \multicolumn{1}{|r}{-0.01140} & 
\multicolumn{1}{r|}{0.00196} & \multicolumn{1}{|r}{-0.00231} & 
\multicolumn{1}{r|}{0.00120} & \multicolumn{1}{|r}{0.00081} & 
\multicolumn{1}{r|}{0.00132} & \multicolumn{1}{|r|}{0.00326} & 
\multicolumn{1}{|r}{0.00123} \\ 
\multicolumn{1}{r}{1/2} & \multicolumn{1}{r|}{1000} & \multicolumn{1}{|r}{
0.01142} & \multicolumn{1}{r|}{0.00075} & \multicolumn{1}{|r}{-0.00782} & 
\multicolumn{1}{r|}{0.00094} & \multicolumn{1}{|r}{-0.00512} & 
\multicolumn{1}{r|}{0.00064} & \multicolumn{1}{|r}{0.00378} & 
\multicolumn{1}{r|}{0.00071} & \multicolumn{1}{|r|}{0.00169} & 
\multicolumn{1}{|r}{0.00070} \\ 
\multicolumn{1}{r}{1/2} & \multicolumn{1}{r|}{2000} & \multicolumn{1}{|r}{
0.01007} & \multicolumn{1}{r|}{0.00048} & \multicolumn{1}{|r}{-0.00772} & 
\multicolumn{1}{r|}{0.00061} & \multicolumn{1}{|r}{-0.00438} & 
\multicolumn{1}{r|}{0.00039} & \multicolumn{1}{|r}{0.00353} & 
\multicolumn{1}{r|}{0.00042} & \multicolumn{1}{|r|}{0.00109} & 
\multicolumn{1}{|r}{0.00043} \\ \hline
\end{tabular}

\bigskip

Table 6. Performance of Feasible Estimator

\begin{tabular}{cc|cc|cc|cc|cc|c|c}
\hline
&  & \multicolumn{2}{|c|}{$\Pr\{Y_0 = 1|\omega\}$} & \multicolumn{2}{|c|}{$\Pr\{Y_0 = 2|\omega\}$}
& \multicolumn{2}{|c|}{$\Pr\{Y_1 = 1|\omega\}$} & \multicolumn{2}{|c|}{$\Pr\{Y_1 = 2|\omega\}$} & 
$F[Y_0|\omega]$ & $F[Y_1|\omega]$ \\ \hline\hline
$\delta $ & $n$ & Bias & MSE & Bias & MSE & Bias & MSE & Bias & MSE & MSE & 
MSE \\ \hline
\multicolumn{1}{r}{1/4} & \multicolumn{1}{r|}{250} & 0.00522 & 0.00360 & 
0.05045 & 0.00785 & -0.04019 & 0.00480 & -0.00071 & 0.00401 & 0.01145 & 
0.00533 \\ 
\multicolumn{1}{r}{1/4} & \multicolumn{1}{r|}{500} & 0.00931 & 0.00237 & 
0.04097 & 0.00489 & -0.03585 & 0.00303 & -0.00602 & 0.00237 & 0.00726 & 
0.00400 \\ 
\multicolumn{1}{r}{1/4} & \multicolumn{1}{r|}{1000} & 0.01099 & 0.00131 & 
0.03946 & 0.00332 & -0.03488 & 0.00215 & -0.00750 & 0.00130 & 0.00464 & 
0.00252 \\ 
\multicolumn{1}{r}{1/4} & \multicolumn{1}{r|}{2000} & 0.01193 & 0.00093 & 
0.04014 & 0.00277 & -0.03345 & 0.00167 & -0.00784 & 0.00071 & 0.00370 & 
0.00171 \\ \hline
\multicolumn{1}{r}{1/3} & \multicolumn{1}{r|}{250} & 0.00340 & 0.00395 & 
0.04905 & 0.00792 & -0.03382 & 0.00415 & -0.00671 & 0.00415 & 0.01187 & 
0.00725 \\ 
\multicolumn{1}{r}{1/3} & \multicolumn{1}{r|}{500} & 0.00409 & 0.00226 & 
0.04718 & 0.00525 & -0.03287 & 0.00278 & -0.00559 & 0.00228 & 0.00751 & 
0.00567 \\ 
\multicolumn{1}{r}{1/3} & \multicolumn{1}{r|}{1000} & 0.00856 & 0.00128 & 
0.03914 & 0.00322 & -0.03292 & 0.00191 & -0.00630 & 0.00125 & 0.00450 & 
0.00486 \\ 
\multicolumn{1}{r}{1/3} & \multicolumn{1}{r|}{2000} & 0.01029 & 0.00086 & 
0.03921 & 0.00261 & -0.02912 & 0.00137 & -0.00819 & 0.00077 & 0.00346 & 
0.00468 \\ \hline
\multicolumn{1}{r}{1/2} & \multicolumn{1}{r|}{250} & 0.00267 & 0.00369 & 
0.04291 & 0.00680 & -0.03508 & 0.00399 & 0.00350 & 0.00367 & 0.01050 & 
0.00643 \\ 
\multicolumn{1}{r}{1/2} & \multicolumn{1}{r|}{500} & 0.00563 & 0.00217 & 
0.04076 & 0.00450 & -0.03205 & 0.00266 & 0.00067 & 0.00215 & 0.00667 & 
0.00417 \\ 
\multicolumn{1}{r}{1/2} & \multicolumn{1}{r|}{1000} & 0.00426 & 0.00123 & 
0.04284 & 0.00357 & -0.02865 & 0.00169 & -0.00184 & 0.00109 & 0.00481 & 
0.00238 \\ 
\multicolumn{1}{r}{1/2} & \multicolumn{1}{r|}{2000} & 0.01084 & 0.00089 & 
0.03538 & 0.00230 & -0.02820 & 0.00128 & -0.00173 & 0.00058 & 0.00319 & 
0.00207 \\ \hline
\end{tabular}

\end{center}
}

\section{Concluding Remarks}

In this paper, we consider identification and estimation of
weakly separable models with endogenous binary treatment. Existing approaches are based on a monotonicity condition, which is violated in models with multiple unobserved idiosyncratic shocks. 
Such models arise in many important empirical settings, including cases where potential outcomes are determined by Roy models, multinomial choice models, or random coefficients with dummy endogenous variables. 
We establish new identification results for these models which are constructive and conducive to estimation procedures. 
A simulation study indicates adequate finite sample performance of our method.

This paper leaves several open questions for future research.
For example, it may not be feasible to locate pairs of $(x,\tilde{x})$ that satisfy the matching criterion perfectly due to limited, say, discrete, support of $X$ or $Z$.
In this case, it remains an open question how or whether the partial identification approach proposed in \citeasnoun{shaikhvytlacil11} can be applied in the current setting where multiple indices in potential outcomes defy any notion of monotonicity.
Besides, our method requires the selection of the number and location of cutoff points, so a data-driven method for selecting these would be useful. Furthermore, the relative efficiency of our proposed estimation approach needs to be explored, perhaps by deriving efficiency bounds for these new classes of models.

\pagebreak 

{\normalsize 
\bibliographystyle{econometrica}
\bibliography{overall}
}

\end{document}